%% file: main.tex
\documentclass{article} 
\usepackage{iclr2024_conference,times}

\input{math_commands.tex}

\usepackage{hyperref}
\usepackage{url,csquotes}
\usepackage{pifont}
\newcommand{\cmark}{\textcolor{blue}{\ding{51}}}%
\newcommand{\xmark}{{\ding{55}}}%
\newcommand{\rem}{m}
\usepackage{hyperref}       
\usepackage{url}            
\usepackage{booktabs}       
\usepackage{amsfonts}       
\usepackage{nicefrac}       
\usepackage{microtype}      
\usepackage{bbm}
\usepackage{subfigure}
\usepackage{amsmath}
\usepackage{multirow, tabu}
\usepackage{array}
\usepackage{tabularx}
\usepackage{rotating}
\usepackage{xurl}
\usepackage{soul}
\usepackage[shortlabels]{enumitem}
\usepackage{makecell}

\usepackage{tikz}
\usetikzlibrary{positioning}
\usetikzlibrary{shapes.geometric, arrows}
\usetikzlibrary{arrows.meta}

\usepackage{amsthm}
\newtheorem{thm}{Theorem}[section]

\newtheorem{proposition}[thm]{Proposition}
\usepackage{array,booktabs,ragged2e}

\newtheorem{definition}{Definition}
\newtheorem{lemma}[thm]{Lemma}

\title{Ranking evaluation metrics from a \\group-theoretic perspective}

\author{Chiara Balestra \thanks{This research was supported by the research training group \enquote{Dataninja} (Trustworthy AI for Seamless Problem Solving: Next Generation Intelligence Joins Robust Data Analysis) funded by the German federal state of North Rhine-Westphalia.} \\
TU Dortmund\\
Germany \\
\And
Andreas Mayr \\
IMBIE, University Hospital of Bonn \\
Germany \\
\And
Emmanuel Müller \\
TU Dortmund\\
Germany \\
}

\iclrfinalcopy 
\begin{document}

\maketitle

\begin{abstract}

    Confronted with the challenge of identifying the most suitable metric to validate the merits of newly proposed models, the decision-making process is anything but straightforward. Given that comparing rankings introduces its own set of formidable challenges and the likely absence of a universal metric applicable to all scenarios, the scenario does not get any better. Furthermore, metrics designed for specific contexts, such as for Recommender Systems, sometimes extend to other domains without a comprehensive grasp of their underlying mechanisms, resulting in unforeseen outcomes and potential misuses. Complicating matters further, distinct metrics may emphasize different aspects of rankings, frequently leading to seemingly contradictory comparisons of model results and hindering the trustworthiness of evaluations.
    
    We unveil these aspects in the domain of ranking evaluation metrics. Firstly, we show instances resulting in inconsistent evaluations, sources of potential mistrust in commonly used metrics; by quantifying the frequency of such disagreements, we prove that these are common in rankings. Afterward, we conceptualize rankings using the mathematical formalism of symmetric groups detaching from possible domains where the metrics have been created; through this approach, we can rigorously and formally establish essential mathematical properties for ranking evaluation metrics, essential for a deeper comprehension of the source of inconsistent evaluations. We conclude with a discussion, connecting our theoretical analysis to the practical applications, highlighting which properties are important in each domain where rankings are commonly evaluated. In conclusion, our analysis sheds light on ranking evaluation metrics, highlighting that inconsistent evaluations should not be seen as a source of mistrust but as the need to carefully choose how to evaluate our models in the future. 
    
\end{abstract}

\section{Introduction}


Evaluating methods is fundamental in any machine learning field, but it is not straightforward finding an \emph{appropriate} evaluation metric that accurately assesses a method's strengths without unfairly biasing comparisons to other approaches. Comparing rankings is particularly challenging: inconsistencies often appear in the produced evaluations, and, notably, inconsistencies diminish users' trust in the methods and evaluations. \par 
Rankings pop up in several domains, from Recommender Systems (RS) and Information Retrieval (IR) techniques~\cite{adomavicius_toward_2005,mogotsi_christopher_2010}, to feature ranking and selection approaches~\cite{khaire2022stability} as well as in (fair) rank aggregation methods~\cite{lin2010rank}. They are looked at as \enquote{interpretable} and easy to understand by humans; as an example, let us think about a scoring system that represents job applicants' characteristics, where a high score is an indicator of a better fit for the position. However, evaluating rankings is not at all that simple, and contradictory evaluations are commonplace. Furthermore, several metrics are appropriately created to evaluate, for example, Recommender Systems or rank aggregation approaches; however, the same metrics are often used in other contexts without a sufficient understanding of their evaluation. \par
First, we show the existence of inconsistencies among most pairs of metrics; given the frequency of inconsistencies' occurrences, evaluations of methods involving rankings are often not trustworthy. We then introduce a list of desirable theoretical properties for ranking evaluation metrics and provide the mathematical framework underpinning each of them; rankings metrics can be easily transferred to functions on mathematical groups, specifically on symmetric groups, thus allowing us to detach from specific machine learning domains. The choice is dictated by the fact that they represent the most general mathematical structure on which we could represent rankings. Profiting from a strong mathematical theoretical interface, we look to answer the question \emph{which mathematical properties are essential for evaluating rankings?} Thus, while most of the existing literature predominantly confines itself rather to narrow, highly specific contexts, our work provides a theoretical framework beyond specific contexts of application. We further assert that almost none of the metrics qualifies as a mathematical distance, as fundamental characteristics are not satisfied. Although symmetric groups offer a broad generalization, we remain mindful of the specific contexts in which the metrics were developed. We highlight the contexts where these properties are particularly valuable in a conclusive discussion. 



\section{Related work}\label{sec:relatedwork}

The literature on ranking evaluation metrics is mostly highly context-specific. Particularly developed for RS and IR evaluation, we find several works exploring the relationships among the metrics~\cite{valcarce_robustness_2018,gunawardana_evaluating_2015,silveira_how_2019}. \cite{herlocker_evaluating_2004} proposes a theoretical division of the metrics for comparing collaborative filtering RS, while \cite{liu2009learning} describes most of the metrics typically used for RS and IR techniques. \cite{jarvelin_cumulated_2002} presents various metrics based on cumulative gain, highlighting their main advantages and drawbacks, \cite{hoyt_unified_2022} introduces a theoretical foundation for rank-based evaluation metrics, and~\cite{amigo_axiomatic_2018} defines a set of properties for IR metrics and shows that none of the existing ones satisfy all the properties proposed. Other works focus on metrics for RS and their intrinsic properties or on ranking metrics for the top-$n$ recommendations~\cite{buckley_retrieval_2004,valcarce_assessing_2020}. 
Real-world applications such as the design of strategies based on customers' feedback and allocation of priorities in R\&D extended the interest in defining distances among rankings where the focus of the problem statement is \emph{rank aggregation}~\cite{10.1145/371920.372165,sculley2007rank}. Examples of similarly scoped works are~\cite{cook1986axiomatic,fligner1986distance}. 
An interesting generalization work is presented by~\cite{diaconis}, that focuses on six metrics on symmetric groups; we find among them, \emph{Kendall's $\tau$} and \emph{Spearmann's $\rho$} while the other considered metrics are rather uncommon in machine learning. The work studies them from a statistical and theoretical perspective and defines some properties, among which the \emph{interpretability}, \emph{tractability}, \emph{sensitivity}, i.e., the ability of one metric to range among the available counter-domain, and \emph{theoretical availability}. In detail, the \emph{interpretability} defined in~\cite{diaconis} discussed whether the metrics measure something humanly tangible, the \emph{tractability} studies the so-called computational complexity in computer science, and the \emph{theoretical availability} asks whether a metric is studied and used enough in the state-of-the-art works. We will reintroduce one of the properties in~\cite{diaconis}, e.g., the \emph{right-invariance}, in our work into our \emph{Robustness property}\color{black}.   \par

The interest in fair and trusted choices of evaluation metrics grows fast in computer science. \cite{tamm_quality_2021} is an example where some of the ranking evaluation metrics are harshly criticized for their comparisons' reliability. In other domains, the state-of-the-art literature started defining essential properties for metrics~\cite{gosgens_good_2021,gosgens_systematic_2021}; we aim to fill the gap for ranking evaluation metrics.

\section{Metrics}\label{sec:metrics}

In the state-of-the-art literature, we find various metrics meant to evaluate ranking in specific contexts. This is the case for most offline Recommender Systems metrics, some evaluation metrics for prediction models, and rank aggregation approaches. Some of these metrics started spreading to other domains following the need to evaluate rankings. However, this is not always a good idea as the domain defines the evaluation's exigencies. In our work, we consider metrics evaluating full rankings that can be easily transferred to adjacent domains and cut out from the analysis all those metrics that require context-specific information, e.g., diversity in Recommender Systems. We refer to the group as ranking evaluation metrics; a complete list is summarized in Table~\ref{tab:listmetrics}. 
\input{tables/listmetrics}
\begin{figure}
    \centering
    \includegraphics[width=0.45\linewidth]{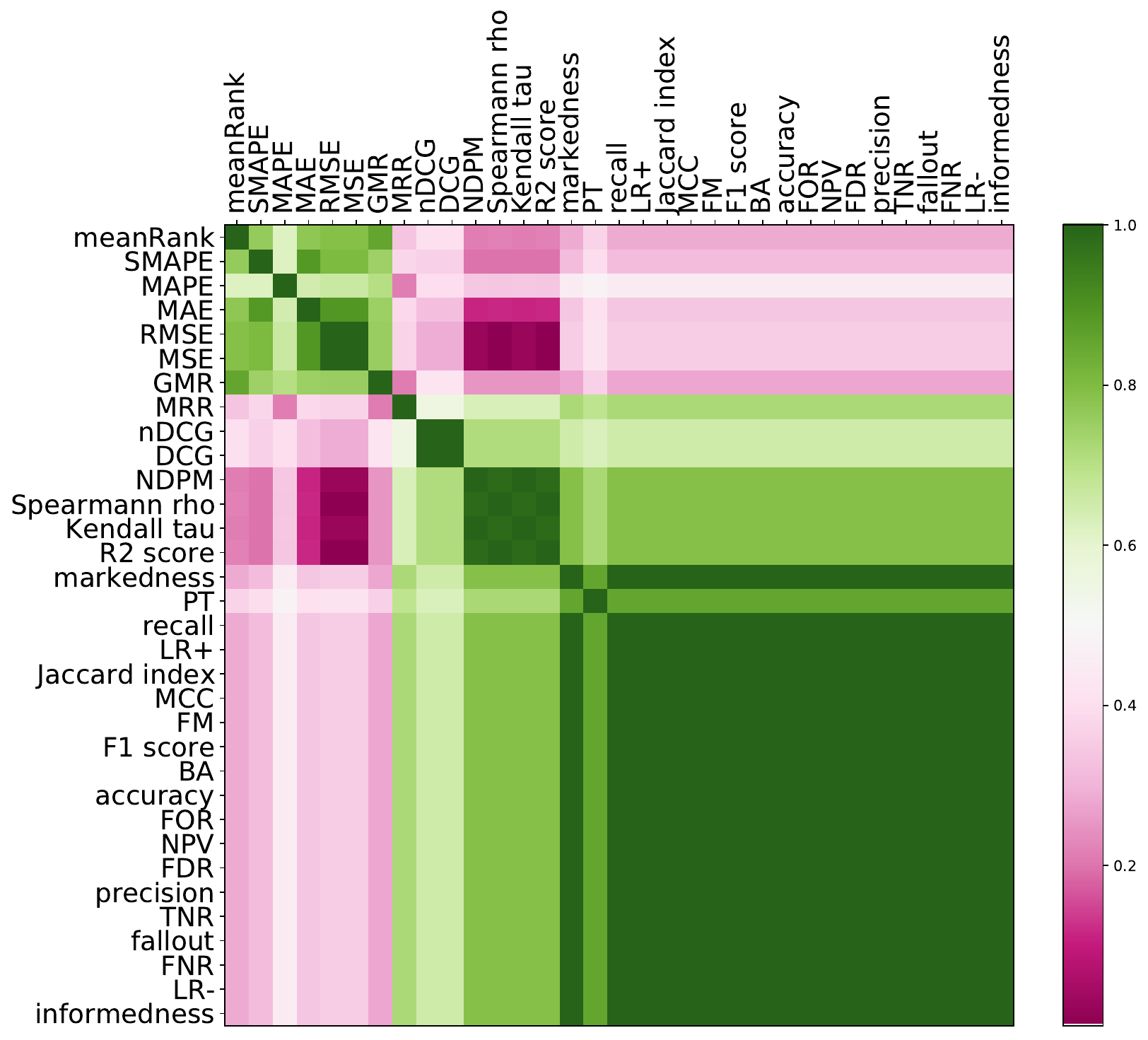}
    \hfill
    \includegraphics[width=0.54\linewidth]{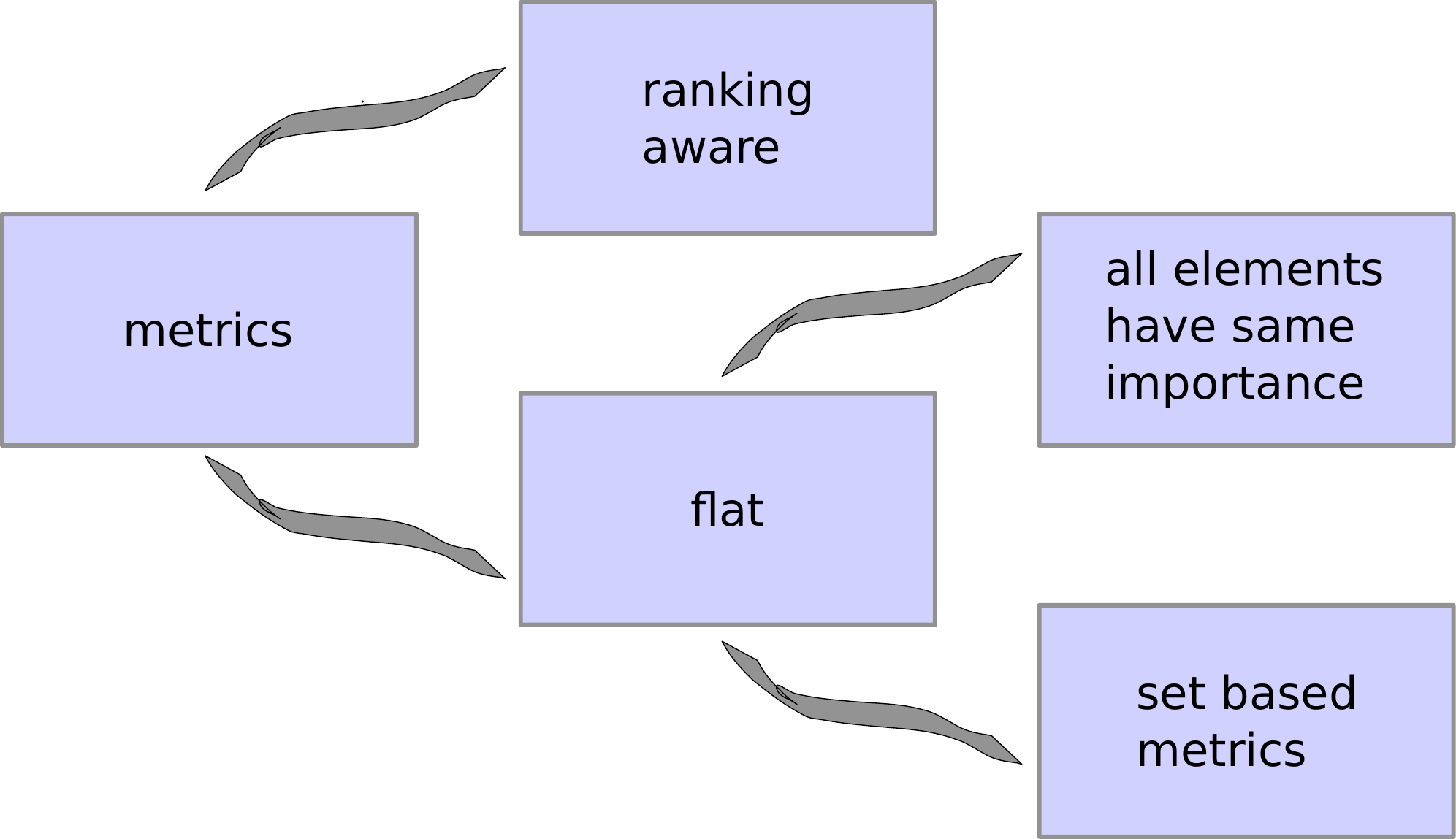}

    \caption{\textbf{On the left:} Heatmap of the agreement ratios among pairs of ranking evaluation metrics. \textbf{On the right:} The theoretical subdivision of the metrics; the }
    \label{fig:heatmap}
\end{figure}
We categorize the ranking evaluation metrics under two different theoretical aspects.
One subdivision derives from their \enquote{awareness} of the position of single items in the rankings: \emph{Ranking aware metrics} satisfy this criterion while \emph{flat metrics} do not. In this second group, we find the \emph{set-based metrics} and the ones assigning equal importance to each position. The subdivision is shown in Figure~\ref{fig:heatmap} on the right. 
From their theoretical definition, we individuate four main groups: \emph{confusion matrix-based CMB metrics} focus on the number of correctly retrieved elements and are essentially set-based metrics; \emph{correlation-based CB metrics} quantify the ordinal association between the two rankings from a statistical perspective; \emph{error-based EB metrics} are often used to analyze the performance of predicting models and are flat metrics assigning equal importance to each position; Finally, \emph{cumulative gain-based CGB metrics} focus on the rankings of the single elements;  additional explanations on how the single metrics have been classified and the definition of the single metrics are available in Appendix~\ref{app:defmetrics}. \color{black}


\section{Ranking evaluation metrics on symmetric groups}\label{sec:theory}


{To generalize the metrics over an abstract structure, we introduce \emph{symmetric groups} $S_n$. Given a finite set $\mathcal{N} = \{1, \ldots, n\}$, the \emph{symmetric group} $S_n$ is the set of bijective functions from $\mathcal{N}$ to $\mathcal{N}$, i.e., the rankings or \emph{permutations} of elements in $\mathcal{N}$; $S_n$ has size $n!$. Permutations are designed with lowercase Greek letters, i.e., $\sigma\in S_n$. Exceptionally, $\text{id}$ indicates the group identity or \emph{identity function}; the \emph{identity function} defines mathematically the supposition that the identical ordering assigns to each 'item' $i$ its position $i$, i.e., that the items' names correspond to their positioning. \color{black} $\sigma(i)$ indicates the position in which item $i$ is sent by $\sigma$ and, given $\sigma,\nu \in S_n$, $\sigma \circ \nu \in S_n$ is a new ranking defined by $\sigma \circ \nu (i) = \sigma(\nu(i)), \forall i \in\{1,\ldots,n\}$; $\circ$ is the group operation and it is not commutative, i.e., generally $\sigma\ \circ\ \nu \neq\nu\ \circ\ \sigma$. 
$\sigma_{|k} = (\sigma(1), \ldots, \sigma(k))$ indicates the ranking of the first $k$ elements; metrics$@k$ consider exclusively the first $k$ ranked elements.} 
{Finally, a \emph{(single) swap} is a permutation $\sigma = (j\ k)\in S_n$, swapping only the two elements $j,k$ in $\mathcal{N}$; \cite{farzad} refers to them as \emph{transpositions}. 
} 

\subsection{Clustering by agreement}\label{sec:clustering}
Our work is mainly justified by the lack of \enquote{consistent} evaluation of rankings when using different metrics.
A \emph{ranking evaluation metric} is a function $\rem: S_n\times S_n \rightarrow \mathbb{R}_+$, taking two permutations as input and returning a real number. In some cases, metrics take only one ranking as input;  we refer to them as \emph{single input metrics}. \color{black}All the given definitions work correspondingly for one-input metrics.  
\begin{definition}
    {Two metrics $\rem_1, \rem_2$ are \emph{non-consistent} if there exists $\sigma ,\mu, \nu \in S_n$ such that the following two conditions hold:}
    \begin{align}\label{eq:inconsistency}
        \begin{array}{ll}
            \rem_1(\text{id}, \sigma) \leq \rem_1(\text{id},\mu)\ \wedge\ \rem_2(\text{id}, \sigma) \leq \rem_2(\text{id},\mu)\\
            \rem_1(\text{id}, \sigma) \leq \rem_1(\text{id},\nu)\ \wedge\ \rem_2(\text{id}, \sigma) > \rem_2(\text{id},\nu)
        \end{array}
    \end{align}
    Otherwise, we say that $\rem_1, \rem_2$ are \emph{consistent}.
\end{definition}

The first line of~\eqref{eq:inconsistency} guarantees that the reversed metric $\tilde \rem_2 = -\rem_2$ is still {non-consistent} \color{black} with $\rem_1$. Proving consistency between two metrics is much trickier than finding three rankings satisfying the inconsistency condition; therefore, rather than classify them, we estimate the degree of inconsistency among pairs of metrics by introducing the \emph{agreement ratio}. 
The coefficient provides an estimate of the extent to which two metrics disagree in the evaluation of rankings over symmetric groups.



\begin{definition}
    For any $\sigma\in S_n$ fixed, the \emph{$\sigma$ agreement ratio} among two ranking evaluation metrics, $\rem_1$ and $\rem_2$ is
    \begin{equation*}
        \text{AR}_{\rem_1,\rem_2}^{\sigma} = \frac{1}{|\mathcal{P}(S_n)|(|\mathcal{P}(S_n)|-1)}\sum_{\mu,\nu \in \mathcal{P}(S_n), \mu\neq\nu } f^{\rem_1,\rem_2}_{\sigma}(\nu, \mu)
    \end{equation*}
    where $f^{\rem_1,\rem_2}_{\sigma}(\nu,\mu) = \mathbbm{1}\{m_1, m_2 \text{ are consistent w.r.t. } \sigma\text{ on the rankings }\mu, \nu\}$ (or equivalently $ f^{\rem_1,\rem_2}_{\sigma}(\nu,\mu) =\mathbbm{1}\{\eqref{eq:inconsistency} \text{ is not satisfied}\}$), $\mathbbm{1}$ is the indicator function  and $\mathcal{P}(S_n)$ is the power set over $S_n$.\color{black}
\end{definition}

As the size of $\mathcal{P}(S_n)$ grows exponentially, we randomly sample a subset $\mathcal{T}$ of $\mathcal{P}(S_n)$ thus obtaining an estimate of the number of inconsistencies existing among two metrics. 
The agreement ratio equals $1$ if $\rem_1$ and $\rem_2$ are consistent and goes to zero with increasing inconsistencies found; furthermore, the agreement ratio is a symmetric metric.  \par 
The color-code heatmap in Figure~\ref{fig:heatmap} highlights, respectively, in green and pink, the existence of a high agreement and disagreement; a partial agreement is represented in white.   
It is visible that similarly theoretically defined metrics as grouped Table~\ref{tab:listmetrics} tend to have an agreement ratio closer to $1$.\color{black} 
The agreement ratio represents an estimate of the number of inconsistencies among metrics; Figure~\ref{fig:heatmap} refers to rankings in $S_{100}$, where $\mathcal{T}$ contains $10000$ random rankings. For CMB metrics, we fixed to $30$ the number of retrieved and relevant elements. We use as reference ranking $\sigma$ the identity function previously defined; however, we would get similar colored heatmaps using other reference rankings.\color{black}



\section{Properties for ranking evaluation metrics}\label{sec:properties}

Most pairs of metrics are affected by frequent inconsistent evaluations (cf. Section~\ref{sec:clustering}). 
We list essential mathematical properties to highlight the peculiarity of each metric and give the chance to properly select one or another based on them for a context-dependent evaluation. 
The properties in question are: (1) \emph{identity of indiscernibles} (IoI); (2) \emph{symmetry} (or \emph{independence from a ground truth}); (3) \emph{robustness} (Type-I and Type-II); (4) \emph{stability} with respect to $k$; (5) \emph{sensitivity} and \emph{width-swap-dependency}; (6) (induced) \emph{distance}. 
Some of them have been defined in other domains,  wee.g.,~\cite{gosgens_systematic_2021,gosgens_good_2021} define the symmetry property for cluster similarity indices and metrics for classification models, \cite{farzad} defines the \enquote{resistance to relabeling} in the context of rank aggregation, \cite{cook1986axiomatic,fligner1986distance} define the importance of constructing distances for partial orderings; will refer to each of them in the respective sections. \color{black}

For each property, we will highlight in which context and why it is important.
Table~\ref{tab:comparison} and Table~\ref{tab:summary} help the reader to keep trace of the mentioned results. The code will be on GitHub upon acceptance \footnote{\url{https://github.com/chiarabales/rankingsEvalMetrics}}. 


\input{tables/identityofindiscrnibles}


\subsection{Identity of indiscernibles}\label{sec:identityofindiscernibles}


Ideally, a metric $\rem$ quantifies how \enquote{close} or \enquote{similar} two rankings $\sigma$ and $\tau$ are. 
However, situations may arise where $\sigma$ and $\tau$ are \enquote{so} similar to be practically indistinguishable by some metrics. 
This effect might be undesired in some fields, such as (fair) rank aggregation, where even small differences, especially in the presence of protected groups, make the difference between fair and unfair rankings. 
\begin{definition}
    A metric $\rem$ satisfies the \emph{identity of indiscernible (IoI) property} if, $\forall\sigma\in S_n$ fixed, the following holds 
    \begin{equation}\label{eq:identityofindiscernibles}
        \rem(\sigma, \tau) = \rem(\sigma, \nu) \Leftrightarrow \tau = \nu, \qquad \forall \tau, \nu \in S_n.
    \end{equation}
\end{definition}
Up to renaming the elements, we can rewrite Equation~\eqref{eq:identityofindiscernibles} as $\rem(\text{id}, \tau) = \rem(\text{id}, \nu) \Leftrightarrow \nu = \tau$ where $\text{id}$ is the usual identity of $S_n$. \par 
Almost all metrics do not satisfy the IoI property;  examples are set-based metrics and metrics$_{@k}$, i.e., where a metric $\rem_{@k}$ evaluates only the top $k$ elements of the rankings \color{black}. Clear examples not satisfying \eqref{eq:identityofindiscernibles} are rankings $\sigma$ that can be written as a disjoint composition of cycles of permutations of elements before and after $k$~\cite{hall2018theory}. Table~\ref{tab:identityofindiscrnibles} illustrates examples for each metric where the IoI is not satisfied. 
It can be proven that
\begin{proposition}\label{prop:DCG_identityofindiscernibles}
    DCG and nDCG satisfy the IoI property.
\end{proposition}
The proof finds place in Appendix~\ref{app:proofs}.
\input{tables/comparison}


\subsection{Symmetry property}\label{sec:symmetry}


Often, guarantees that the evaluation is symmetric with respect to input items are desirable~\cite{gosgens_good_2021,gosgens_systematic_2021}, particularly when the interest is in having a sort of mathematical distance, e.g., for rank aggregation. \color{black}
However, as usual, the context rules the need for a symmetric evaluation. The symmetry property studies whether the metric's evaluation is independent of the order in which the rankings are compared.  In RS and IR, the common presence of a \enquote{ground truth order} makes the symmetric property impossible.

\begin{definition}
    A metric $\rem:S_n\times S_n \rightarrow \mathbb{R}$ is \emph{symmetric} if 
    \begin{equation}
        \rem(\sigma,\nu) = \rem(\nu,\sigma),\qquad\forall\sigma, \nu \in S_n.
    \end{equation}
\end{definition}



\subsection{Robustness}\label{sec:robustness}

The IoI property studies whether metrics can distinguish rankings, regardless of their similarity. 
On the other side, the similarities among rankings should be projected on the evaluations: 
Small differences in rankings should result in small differences in the evaluation scores. 
Under the assumption that a single swap represents a small difference between two rankings, the \emph{Type I  robustness} property assesses how sensitive a ranking evaluation metric is to single swaps in the compared rankings.
\begin{definition}

     A metric $\rem$ is \emph{Type I Robust} if a \emph{single swap} in one of the rankings implies small changes in its evaluation, i.e.,  
     \begin{equation}\label{eq:swap}
            |\rem(\sigma, \nu) - \rem(\sigma, \nu \circ (i\ j)| < \epsilon, \qquad \forall \sigma, \nu \in S_n.
        \end{equation}
\end{definition}
        
We compute the average of the results of~\eqref{eq:swap} evaluated on a set $\mathcal{I} \subseteq {S_n}\times {S_n}$ of $1000$ different randomly drawn pairs of rankings in $S_{100}$, i.e., $\sum_{(\sigma,\nu)\in \mathcal{I}, (i, j)\in \{1,\cdots,n\}^2}|\rem(\sigma, \nu) - \rem(\sigma, \nu \circ (i\ j)|$ and round it to two decimal numbers. We state that the metric satisfies the Type I Robustness if the resulting average is $0$.

For completeness, we define a second type of robustness that studies the effect of renaming the items in the rankings. 
\cite{diaconis} mentions Type II Robustness as \enquote{right-invariance} and \cite{farzad} as \enquote{resistance to item relabeling}. \color{black}
\begin{definition}
     A metric is \emph{Type II Robust} if it is an invariant w.r.t. the composition of permutations, i.e., it holds 
     \begin{equation*}
         \rem(\mu, \sigma) = \rem(\mu \circ \nu, \sigma \circ \nu), \forall \mu, \sigma, \nu \in S_n.
 \end{equation*}
\end{definition}
\color{black}\emph{Type II Robustness} property investigates whether applying the same change in both rankings affects the evaluation. The property is essential in contexts where the numbers appearing in the rankings have to be considered as proper \enquote{items} or \enquote{items' names}; this is often the case in rank aggregation approaches, Recommender Systems, and Information Retrieval techniques. 
However, it does not apply when dealing with importance scores. 
We claim that
\begin{proposition}\label{prop:robustness3}
    \emph{MSE}, \emph{RMSE}, \emph{MAE}, \emph{MAPE}, \emph{R$^2$ score}, \emph{Kendall's $\tau$ score} and \emph{Spearmann's $\rho$} are the only considered metrics satisfying the Type II Robustness.
\end{proposition}
The proof derives directly from their definitions (see Appendix~\ref{app:proofs}).


\subsection{Sensitivity}\label{sec:sensitivity}


The sensitivity property is valuable for a metric, particularly in the case of Recommender Systems and Information Retrieval, where high dimensional rankings may not be fully explored. 
Under the assumption that a full exploration of the rankings is not possible, sensitive metrics assign more weight to the first part of the rankings, considering whether the first $k$ items are \enquote{correctly} ranked. 
Mathematically, we define the \emph{width} of a swap $(i\ j)\in S_n$ being the quantity $|i-j|$. \color{black}
\begin{definition}
 Given $i <j <k<l \in \{1, \ldots,n\}$ and $(i\ j), (l\ k)$ having the same width. A ranking evaluation metric $\rem$ is \emph{sensitive} if $\exists\sigma\in S_n$ such that it holds $\rem((i\ j)\circ \sigma) \neq \rem ((k\ l)\circ \sigma)$.
\end{definition} 
\color{black}As the evaluation of the property is far from easy, we introduce the \emph{width swap dependency}, formalizing a property that prevents the metrics from being sensitive. 
\begin{definition}
     Given a swap $(i\ j)\in S_n$ and $|i-j|$ its \emph{width}, $\rem$ is \emph{width swap dependent} (WSD) if it evaluates swaps with the same width equally, i.e., $\rem((i\ j)) = \rem ((k\ l))$ if $|i-j| = |k-l|$ holds; otherwise, it is called \emph{non-width swap dependent}.
\end{definition}
\color{black}The WSD property cuts out some of the metrics from being sensitive.
From their definitions, it can be proven that 
\begin{lemma}\label{lem:tau_widthswapdependent}
    Kendall's $\tau$, Spearmann $\rho$, NDPM  are \emph{width swap independent}.
\end{lemma}
The proof finds place in Appendix~\ref{app:proofs}. For the other metrics,it is trivial to find pairs of disjoint swaps had different effects in the final evaluation when happening at various positions within the rankings. 


\subsection{Stability}\label{sec:stability}


We introduce the stability property for those metrics that can be applied on \enquote{rankings $@k$}. 
We recall that a ranking at $k$ is the ranking of the items in the first $k$ positions.
To evaluate rankings $@k$, it is essential that the difference between evaluations \enquote{$@k$} and \enquote{$@k+1$} is not significant, i.e., that the choice of $k$ does not highly impact the result; this guarantees a trustworthy evaluation. 
\begin{definition}  
     A ranking evaluation metric $\rem$ is \emph{stable} if, for any two rankings $\sigma,\nu\in S_n$, it holds 
     \begin{equation}
        \label{eq:stability}
        \left|\rem_{@k-1}(\sigma,\nu) - \rem_{@k}(\sigma, \nu)\right| < \epsilon_k \end{equation}
      with $\epsilon_k$ small. Moreover, the sequence $\{\epsilon_k\}_k$ satisfies $\lim_{k\to n}\epsilon_k = 0$.
\end{definition}
The property is again essential for extremely long rankings and for contexts where rankings are not fully explored.
We evaluate the stability by randomly drawing $1000$ pairs of rankings in $S_{100}$, computing the absolute differences of~\eqref{eq:stability}, and counting the number of times that~\eqref{eq:stability} holds with $\epsilon_k = \frac{1}{k}$. We state that a metric is stable if the criterion is satisfied in at least $97.5 \%$ of the cases. \color{black} 


\subsection{Distance}\label{sec:distance}

In mathematics, the terms metric and distance are synonyms. 
However, when it comes to evaluation metrics, most of them are not \enquote{distances} on $S_n$ in the mathematical sense. 
Whether a metric is a mathematical distance or not is often insignificant for the final evaluations; however, being aware of this fundamental mathematical difference can avoid incomprehension and misuses. 
\begin{definition}\label{def:distance}
    A \emph{distance} on a set $X$ is a function $f_{\rem}: X \times X \rightarrow [0, \infty): (x,y) \mapsto f_{\rem}(x,y) \in \mathbb{R}_+$ that, for all $x,y,z \in X$, satisfies: 
    \begin{enumerate}
        \item $f_{\rem}(x,y) =0 \Leftrightarrow  x=y$,
        \item the \emph{positive definiteness}, i.e., $f_{\rem}(\sigma, \nu) \geq 0, \forall \sigma, \nu \in X$,
        \item the \emph{symmetry} property and
        \item the \emph{triangle inequality}, i.e., $f_{\rem}(x,y) \leq f_{\rem}(x,z)+f_{\rem}(z,y)$.
    \end{enumerate}
\end{definition}
Some ranking evaluation metrics are distances; In~\cite{farzad,diaconis}, it is proven that Kendall's $\tau$ is a distance.
However, a ranking evaluation metric that does not satisfy some of the properties mentioned in Definition~\ref{def:distance} is not a distance.

We investigate if we can induce distances from single input metrics. 
Given a metric $\rem:S_n \rightarrow \mathbb{R}$, we consider two options as potential induced distances, i.e., $f_{\rem}(\sigma, \nu) = \rem(\sigma) - \rem(\nu)$ or $\tilde f_{\rem}(\sigma, \nu) = | \rem(\sigma) - \rem(\nu) | $. DCG and nDCG are the only two metrics satisfying the IoI property that, for metrics with one unique argument, is equivalent to Property (1) for $f_{\rem}$. 
We can easily prove that 
\begin{proposition}\label{prop:distance}
    $f_{\rem}$ is \emph{not} a distance while $\tilde f_{\rem}$ is a distance with the IoI property, where $\rem$ is either DCG or nDCG.
\end{proposition}
The formal proof can be found in Appendix~\ref{app:proofs}. 

\section{Are the metrics interpretable? Thoughts over maximal and minimal agreement properties }\label{sec:interpretability}


Given the importance of trust, fairness, and explainability for machine learning methods, one could then ask how \enquote{interpretable} the scores assigned by the metrics are. We first need some definitions.
\begin{definition}\label{def:maximalagreement}
    A ranking evaluation metric $\rem$ is said to satisfy the \emph{maximal agreement property} if (a) $\rem(\sigma, \sigma) = \rem_{\max}, \forall\sigma\in S_n$ and (b) $\rem(\sigma, \nu) \leq \rem_{\max}, \forall\nu, \sigma\in S_n$. We say that $\rem$ is \emph{lower-bounded} if it exists a real number $\rem_{\min}$ such that $\rem(\sigma, \nu) \geq \rem_{\min}, \forall\nu, \sigma\in S_n$. An evaluation metric that admits a lower bound is said to satisfy the \emph{minimal agreement} property.
\end{definition}
For a metric to be \enquote{interpretable} we expect that
\begin{enumerate}
    \item each ranking is maximally similar to itself and, given $n\in \mathbb{N}$, this value is constant, i.e., $\rem(\sigma,\sigma) = \rem_{\max}, \forall \sigma \in S_n$ and $\forall n$; 
    \item $\rem$ satisfies the maximal agreement property;
    \item there exists a lower bound $\rem_{\min}$ for any possible pair of rankings, i.e., $\rem(\sigma,\mu) \geq \rem_{\min}, \forall\sigma, \mu \in S_n$.
\end{enumerate}

Exemplary is the Kendall's $\tau$ metric which satisfies $\rem(\sigma, \sigma) =1 $ and $\rem(\sigma, \mu)\in [-1,1]$ for all $\sigma, \mu\in S_n$. \color{black}
The maximal agreement property says that each ranking is maximally similar to itself, and no other ranking can achieve a higher score than $\rem_{\max}$; furthermore, ideally, $\rem_{\max}$ is independent of the length of the rankings. Properties (1) and (2) imply that a ranking evaluation metric is a monotone increasing function of the similarity of two rankings: the more similar two rankings are, the higher the score they get when evaluated using an \enquote{interpretable} metric. 
Having that $\rem_{\max}$ is independent of $n$ is a necessary condition for having an evaluation of rankings independent of $n$. \color{black} 
However, this is hardly satisfied by any metrics, and only after introducing a normalization score do the metrics satisfy the requirement. 
Furthermore, the lowest scores are assigned by some metrics to maximally similar pairs of rankings, e.g., error-based metrics. 
The only metrics, among the ones considered in this paper, automatically satisfying this property are Kendall's $\tau$ score and Spearmann $\rho$. \par 
A ranking evaluation metric satisfying the maximal agreement property is also \emph{upper-bounded}. 
For the sake of interpretability, we could check whether a metric $\rem$ satisfies $\rem(\rho^{-1},\rho)= \rem_{\min}$ where $\rho^{-1}$ indicates the inverse ranking. 
How do we define the inverse of a ranking? 
Kendall's $\tau$ satisfies this property, given that the inverse of one ranking $\sigma$ is the ranking $\tau$ assigning the highest position to the last element of the ranking $\sigma$; however, this does not correspond with the inverse of the ranking in the symmetric group. 
Assessing whether metrics for permutations are humanly interpretable is not new and has already been discussed in~\cite{diaconis}. 
However, then, as well as now, the concept of interpretability lacks a unified definition. 
Thus, we leave this section open and do not argue further on the interpretability of the considered metrics.

 \input{tables/summary}

\section{Discussion}

We explored metrics for comparing and evaluating rankings and analyzed their theoretical properties. 
All the mentioned metrics are widely used in the literature to evaluate Recommender Systems, Information Retrieval, feature ranking, rank aggregation methods, and items' score assignments. 
Each property is highly desirable in some contexts and less in others. 
The IoI property is desirable in highly sensitive evaluations, where detecting tiny differences among rankings is essential; fair ranking aggregation is an example, where swapping items can make the difference between fair and unfair rankings. 
Conversely, robustness ensures that small changes influence the evaluations proportionately in a one-to-one fashion. 
A metric that satisfies both the IoI and the robustness properties ensures contemporaneously that small changes are not overlooked but do not significantly impact the evaluations. 
The symmetry property ensures that the input rankings have an equal role in the evaluations. This is essential in most domains unless ground truth ranking is available. 
Note that non-symmetric metrics are also not distances. 
Rank aggregation is again an example of use for the symmetry property, where the consensus ranking is directly compared with the original rankings provided. 
Sensitivity is crucial when rankings are not fully explored. This is often the case for Recommender Systems and Information Retrieval techniques' evaluations. 
With the same applicability, the stability property ensures trustworthiness in evaluations $@k$, which is again highly relevant for Recommender Systems and Information Retrieval techniques. 
To assure stable evaluations, we recommend considering evaluating the impact $@k$ and $@(k+i)$ with $i$ arbitrarily chosen, in particular when $k <<n$. Finally, the distance property is defined to complete the proposed analysis and highlights the chance that mathematical terms are misused in machine learning contexts.
Table~\ref{tab:summary} summarizes the properties' descriptions and application domains.

\section{Conclusion}
Throughout the paper, we explored metrics widely used in the literature to evaluate Recommender Systems, Information Retrieval, feature selection, and rank aggregation methods; rankings are the common output of all these methods. 
We observed a common presence of {non-consistent} evaluations of rankings, deriving from the different definitions of the ranking evaluation metrics. 
Focusing on a mathematical perspective and viewing rankings as elements of symmetric groups and the metrics as functions defined over mathematical groups, we list a set of well-founded mathematical properties for ranking evaluation metrics. 
The differences among metrics are highlighted by the differences in the satisfiability of the properties, thus grounding the reasons for the inconsistencies in the evaluations.
As each property is highly desirable in some contexts and less in others, we summarize the obtained insights that can be of immediate use when looking for an appropriate metric for a specific domain.

\bibliography{main} 
\bibliographystyle{plain}

\appendix
\newpage

\section{Proofs}\label{app:proofs}

\begin{proof}[Proof of Proposition~\ref{prop:DCG_identityofindiscernibles}]
    As DCG and nDCG differ only for a constant multiplicative factor, we prove the claim only for DCG. As we deal with pure rankings on symmetric groups, we use the convention $\text{rel}_i = \sigma(i)$ representing a rescaling of the relevance score to distinguished integer numbers; Given $\sigma\in S_n$, we use the following definition $\text{DCG}(\sigma) = \sum_{i = 1}^n \frac{\sigma(i)}{\log_2(i+1)}$. \par 
    The goal is proving that for any $\sigma_1, \sigma_2\in S_n$, $\text{DCG}(\sigma_1) = \text{DCG}(\sigma_2) \Leftrightarrow \sigma_1 = \sigma_2$. Without loss of generality, we prove that $\text{DCG}(\text{id}) = \text{DCG}(\sigma) \Leftrightarrow\sigma = \text{id}$ for any $\sigma\in S_n$, i.e., $\sum_{i = 1}^n \frac{i -\sigma(i)}{\log_2(i+1)} =0$. As this is not straightforward, we prove instead a stronger version
    \begin{equation}\label{eq:thesisInequality}
        \sum_{i = 1}^n \frac{i -\sigma(i)}{\log_2(i+1)} < 0 \Leftrightarrow \sigma \neq \text{id}\in S_n.
    \end{equation}
    We base our proof on induction over $n$. \par  \vspace{0.2cm}
    \textbf{Base case:} The base case $n=2$ is trivial as $S_2=\{\text{id}, \sigma = (1\ 2)\}$; in particular, $\text{DCG}(\text{id}) = 0$ while 
    \begin{equation*}
        \text{DCG}(\sigma) = \frac{1-\sigma(1)}{\log_22} + \frac{2-\sigma(2)}{\log_23} =-\frac{1}{\log_22} + \frac{1}{\log_23} < 0
    \end{equation*} \vspace{0.2cm}
    \textbf{Inductive case:} The claim holds for $n-1$ and we prove it for $n$; consider $\sigma\in S_n$. We distinguish two cases. \vspace{0.2cm}\\
    \textbf{One element is fixed by $\sigma$:} Up to renaming the elements, we suppose that $n$ is fixed by $\sigma$, i.e., $\sigma(n) = n$. Given $n, k\in \mathbb{N}$, we can construct an immersion $i_{n,k}:\sigma \in S_n \mapsto i_{n,k}(\sigma)\in S_{n+k}$ of $S_n$ in $S_{n+k}$, such that $i_{n,k}(\sigma)(j) = \sigma(j)$ if $ j \leq n$ otherwise $i_{n,k}(\sigma)(j) = j$; $i_{n,k}$ is injective and surjective on $A = \{\sigma \in S_{n+k}\mid \sigma(j) = j, \forall j > n+k\}$ and $\sigma $ fixes $n$, $\sigma$ belongs to $S_{n-1}$ (as the counter-image of $i_{n,1}$). Therefore, the claim holds. \vspace{0.2cm}\\
    \textbf{No element is fixed by $\sigma$:} It holds $\sigma(n) \neq n$ and we can rewrite $\sigma$ as the composition of two permutations, i.e., $\sigma = \tau \circ \mu$ such that $\tau =(j\ n)$ for some fixed $j$ and $\mu$ such that $ \mu(s)= \sigma(s)$ if $ s\neq n, k^*$, $ \mu(s) = j$ {if} $s = k^*$ and $ \mu(s) = n $ if $ s= n$ where we named $k^*=\mu^{-1}(j) = \sigma^{-1}(n)$. We can now rewrite $\sigma$ in terms of $\tau \circ \mu$;
        \begin{align*}
            &\sum_{i = 1}^n \frac{i -\sigma(i)}{\log_2(i+1)} =  \\
            &\sum_{i = 1, i \neq k^*}^{n-1} \frac{i -\sigma(i)}{\log_2(i+1)} + \frac{k^* -\sigma(k^*)}{\log_2(k^*+1)} + \frac{n -\sigma(n)}{\log_2(n+1)} = \\
            & \sum_{i = 1, i \neq k^*}^{n-1} \frac{i -\mu(i)}{\log_2(i+1)} + \frac{k^* -\tau \circ\mu(k^*)}{\log_2(k^*+1)} + \\&\frac{n -\sigma(n)}{\log_2(n+1)} +\frac{k^* -\mu(k^*)}{\log_2(k^*+1)} - \frac{k^* -\mu(k^*)}{\log_2(k^*+1)} = \\
            & \sum_{i = 1}^{n-1} \frac{i -\mu(i)}{\log_2(i+1)} + \frac{k^* -\tau (j)}{\log_2(k^*+1)} + \\&\frac{n -\sigma(n)}{\log_2(n+1)}  - \frac{k^* -\mu(k^*)}{\log_2(k^*+1)} 
        \end{align*}
        $ \sum_{i = 1}^{n-1} \frac{i -\mu(i)}{\log_2(i+1)}$ is negative for the inductive hypothesis and we assume that $\mu\neq \text{id}\in S_{n-1}$. By substituting $\sigma = \tau \circ \mu $, we conclude the proof if we can upper bound their sum with $0$. 
        \begin{align*}
            &\frac{k^* -\tau (j)}{\log_2(k^*+1)} + \frac{n -\sigma(n)}{\log_2(n+1)} - \frac{k^* -\mu(k^*)}{\log_2(k^*+1)} = \\&\frac{k^* -n - (k^* -\mu(k^*))}{\log_2(k^*+1)} + \frac{n -\sigma(n)}{\log_2(n+1)} = \\
            &\frac{\mu(k^*) -n}{\log_2(k^*+1)} + \frac{n -\sigma(n)}{\log_2(n+1)} <  \\& \frac{\mu(k^*) -n}{\log_2(k^*+1)} + \frac{n -\sigma(n)}{\log_2(k^*+1)} = 0
        \end{align*}
        where we used $\log_2(n+1) > \log_2(k^*+1) $, $\sigma(n) = \tau \circ\mu(n) = \tau(n) = j$ {and} $\mu(k^*) = j$. Thus, the claim is proved for $\mu \neq \text{id}$. In the case $\mu = \text{id}$: Then it holds $\sigma = \tau$ and $\text{DCG}(\sigma)$ reads
    \begin{align*}
        &\sum_{i = 1}^n \frac{i -\sigma(i)}{\log_2(i+1)} = \sum_{i = 1}^n \frac{i -\tau(i)}{\log_2(i+1)} =\\& \frac{j -\tau(j)}{\log_2(j+1)} + \frac{n -\tau(n)}{\log_2(n+1)} = \\
        &  \frac{j -n}{\log_2(j+1)} + \frac{n - j}{\log_2(n+1)} < \frac{j -n +(n-j)}{\log_2(j+1)} = 0
    \end{align*}
    This concludes the proof.
\end{proof}


\begin{proof}[Proof of Proposition~\ref{prop:robustness3}]
    \textbf{\emph{MSE}, \emph{RMSE}, \emph{MAE}, \emph{MAPE}, \emph{R$^2$ score}:} decomposing the sum in the definition of $MSE(\sigma \circ (j\ k), \nu \circ (j\ k))$ among addends involving $k$ or $j$ and others, it is easy to get to $MSE(\sigma , \nu)$. Similarly, for the other metrics. \textbf{Kendall's $\tau$:} it is enough to note that the number of discordant and concordant pairs does not change when applying a swap to both the rankings $\sigma$ and $\nu$. \textbf{Spearmann's $\rho$:} similarly to the case of the error based metric, we decompose the sum defining the Spearmann's $\rho$ in elements involving $j$ and $k$ and others; manipulating the definition, we eventually get the thesis. 
    \textbf{Unicity:} For all the other metrics, finding pairs of rankings providing counterexamples is trivial. For cumulative gain based metrics, the swaps change the association between the position in the ranking and the relevance score. For confusion matrix based metrics, swaps change both the set of relevant and retrieved elements (but not equally); Thus, the evaluation is different after applying swaps in both rankings.
\end{proof}


\begin{proof}[Proof of Proposition~\ref{lem:tau_widthswapdependent}]
    \textbf{Spearman's rank correlation coefficient} has an equivalent formulation that depends only on the differences $d_i = \sigma(i) - \nu(i)$; The fact that the elements appearing in the ranking are all distinct implies the WSD property directly. \par
    To prove the claim for \textbf{Kendall's $\tau$} (NDPM is similar), we fix an arbitrary $n$ and a swap $(i\ j)\in S_n$ of width $d$. We proceed by induction on $d$ and prove that Kendall's $\tau$ is based only on $d$, independently from $i$ and $j$. If $d = 1$, then the swap is of the form $(i\ \ i+1)$; in this case, the number of concordant pairs is $\binom{n}{2}-1$, and the only discordant pair is given by $(i\ i+1)$. Recalling the definition of Kendall's $\tau$, we want to prove that 
    \begin{align*}
        K_{\tau} & = \frac{|\{\text{concordant pairs}\}|-|\{\text{discordant pairs}\}|}{\binom{n}{2};} = \\&\frac{\binom{n}{2}-4|i-j|+2}{\binom{n}{2}}.
    \end{align*}
    This holds for $d=1$ as $K_{\tau}(id,(i\ j)) = \frac{\binom{n}{2}-1+ \left(\binom{n}{2}-\left(\binom{n}{2}-1\right)\right)}{\binom{n}{2}}= \frac{\binom{n}{2}-2}{\binom{n}{2}}$. We now suppose that it holds for $d$ and prove it for $d+1$; the number of discordant pairs in a swap of length $d+1$ equals the number of elements that are not anymore concordant with $i$, i.e., $d+1$, plus the number of elements that are not anymore concordant with $j$ minus $1$, i.e., $d$; summing up we get 
    \begin{align*}
        &K_{\tau}(id,(i\ j)) = \\ &\frac{\binom{n}{2}-(2d+1)+ \left(\binom{n}{2}-\left(\binom{n}{2}-(2d+1)\right)\right)}{\binom{n}{2}}= \\ &
         \frac{\binom{n}{2}-4(d+1)+2}{\binom{n}{2}}.
    \end{align*} We conclude that Kendall's $\tau$ is \emph{width-swap-dependent}. 
\end{proof}


\begin{proof}[Proof of Proposition~\ref{prop:distance}]
We must prove the three properties defining a distance for $m = \text{DCG}$ (similar to nDCG, that differs only by a multiplicative factor). \par
\textbf{Identity Property:} Proposition~\ref{prop:DCG_identityofindiscernibles} states that DCG satisfies the IoI property. Furthermore, it follows that $f_{\rem}(\sigma,\nu) = 0 \Leftrightarrow \sigma = \nu$; Similarly, $\tilde f_{\rem}(\sigma,\nu) = 0 \Leftrightarrow \nu = \sigma$.\par
\textbf{Symmetry property:} It is easy to find pairs of permutations $\sigma, \nu \in S_n$ such that $f_{DCG}(\nu, \sigma) = f_{DCG}(\sigma, \nu)$; In particular, $f_{DCG}$ satisfy the anti-symmetric property, i.e., $f_{DCG}(\nu, \sigma) =\ DCG(\nu) - DCG(\sigma) = - [DCG(\sigma) - DCG(\nu)] = - f_{DCG}(\sigma, \nu)$. On the other hand, $ \tilde f_{DCG}$ satisfies the symmetry property.\par

\textbf{Triangle inequality:} The triangle inequality property is satisfied if $\forall\nu, \sigma, \mu\in S_n$ holds $f_{DCG}(\sigma, \mu)\leq f_{DCG}(\sigma, \nu) + f_{DCG}(\nu, \mu)$. Expanding the formula of DCG we get
\begin{align*}
    f_{DCG}(\mu, \sigma)& \  = DCG(\mu) - DCG(\sigma) = \\&DCG(\mu) - DCG(\nu) +DCG(\nu) - DCG(\sigma) = \\&f_{DCG}(\mu, \nu) + f_{DCG}(\nu, \sigma);
\end{align*}
The equality holds $\forall\nu, \sigma, \mu\in S_n$; for $\tilde f_{DCG}$, the property still holds with the inequality:
\begin{align*}
    &\tilde f_{DCG}(\mu, \sigma) =\ |DCG(\mu) - DCG(\sigma)| =\\ & |DCG(\mu) - DCG(\nu) +DCG(\nu) - DCG(\sigma)| \leq\\
   \leq & |DCG(\mu) - DCG(\nu)| + |DCG(\nu) - DCG(\sigma)|  =\\ & \tilde f_{DCG}(\mu, \nu) + \tilde f_{DCG}(\nu, \sigma).
\end{align*}\par
\textbf{Positive definiteness:} $\tilde f_{DCG}$ is defined as an absolute value; the claim obviously holds. Instead, $f_{DCG}$ can assume both positive and negative values. This concludes the proof.
\end{proof}

\section{Metrics' definitions} \label{app:defmetrics}
    We give some insights on the metrics we mentioned and analyzed throughout the paper; we also properly define here some of the terms and metrics we used.
    \subsection{Confusion matrix based metrics} 
    Consider a set $N$ of $n$ elements, a subset ${R}\subset N$ of the relevant elements and a subset ${S}\subset N$ of the retrieved elements; the confusion matrix, i.e., a $C = \mathbb{N}^{2\times2}$, is defined such that $C_{1,1} = |S\cap R|$, $C_{1,2} = |S\setminus R|$, $C_{2,1} = |R\setminus S|$ and $C_{2,2} = n -|S\cup R|$. Each metric in this group is defined on the sizes of intersections, unions or differences among the sets $R,S$ and $N$. Given two rankings $\sigma, \tau\in S_n$ and two natural numbers $j,k<n$, we can define the set of relevant elements ${R}$ being the first $j$ths ranked elements by $\sigma$ and the set of  retrieved elements being the first $k$ths ranked elements by $\tau$, i.e., ${R} =\text{set}\left(\sigma_{|j}\right)$ and ${S} =\text{set}\left(\tau_{|k}\right)$. \par 
    \textbf{General properties.} Defined only on set based quantities, the ordering of the elements appearing before and after $j$ (or $k$) is irrelevant. Some of these metrics represent a powerful tool for evaluating and comparing rankings. Their strength is well founded on the simplicity and interpretability of the definitions; however, one should consider more sophisticated evaluation metrics when the interest is in the rankings rather than the ability to retrieve the relevant elements. 
			
	We shortly define the used metrics and, for sake of readability, we drop the notation $(\sigma,\tau)$; we consider $k = j$ throughout the manuscript. The \emph{precision} represents the fraction of the number of retrieved elements that are relevant, i.e., $\text{precision}= \frac{| {R}\cap S|}{|S|}$. The \emph{recall} represents the fraction of relevant elements successfully retrieved, i.e., $	\text{recall}= \frac{|{R}\cap S|}{|{R} |}$; It is often referred to as \emph{sensitivity}. The \emph{Fallout} represents the proportion of non-relevant elements that are retrieved, i.e., $\text{fallout}=\frac{| (N\setminus {R})\cap S|}{|N\setminus {R}|}$. The \emph{F-score} is the harmonic mean of precision and recall where precision and recall can also be not be evenly weighted $F_{\beta}= \frac{(1+\beta^2)\cdot(\text{precision}+\text{recall})}{\beta^2\text{precision}+\text{recall}}$; 				if $\beta = 1$ then {precision} and {recall} are evenly weighted and we refer to it as F1-score. The \emph{accuracy} is defined as $\text{ACC} = \frac{|S\cap R|+ n-|S\cup R|}{n}$. The \emph{Jaccard index} is defined as $\text{Jaccard} =\frac{|S\cap R|}{|S\cup R|}$. The \emph{Matthews correlation coefficient (MCC)} is defined as  $\text{MCC}=\frac{|R\cap S|(n-|S\cup R|)- |S\setminus R||R\setminus S|}{|S||R|(n-|R|)(n-|S|)}$. Given the quantities $\text{TPR}= \frac{|R\cap S|}{|R|}$, $\text{TNR}=\frac{n-|R\cup S|}{n-|R|}$, $\text{FNR}=1-\text{TPR}$, $\text{FPR}= 1-\text{TNR}$, $\text{FDR}=\frac{|S\setminus R|} {|S|}$ and $\text{NPV}=\frac{n-|R\cup S|} {n-|S|}$, we can also define the \emph{informedness}, i.e., $\text{informedness}= \text{TPR} + \text{TNR}-1$, the \emph{markedness}, i.e., $\text{markedness} =1-\text{FDR}+\text{NPV}-1$, the \emph{false omission rate {FOR}}, i.e., $\text{FOR}=1-\text{NPV}$, the \emph{prevalence threshold {PT}}, i.e., $\text{PT} = \frac{\sqrt{TPR\cdot FPR}-FPR}{TPR-FPR}$, the \emph{Fowlkes–Mallows index FM}, i.e.,  $\text{FM} =\sqrt{(1-FDR)TPR}$, the \emph{balanced accuracy} {BA}, i.e., $\text{BA}=\frac{\text{TPR+TNR}}{2}$, and finally the \emph{Positive likelihood ratio LR+}, i.e., $\text{LR+} = \frac{\text{TPR}}{1-\text{TNR}}$ the \emph{Negative likelihood ratio LR-}, i.e., $\text{LR-} =\frac{1-TPR}{TNR} $.


			\subsection{Correlation measures}
			They explicitly rely on the correlation among the two rankings and are often used in statistical applications. In contrast to confusion matrix based metrics, they consider all the length of the rankings. \emph{Kendall's $\tau$ coefficient} and \emph{Spearmann $\rho$} consider the permutation of the elements over arrays of length $n$~\cite{kendall_new_1938}. The \emph{Kendall's $\tau$ coefficient} is based on the definition of concordant and discordant couples (two elements $i,j$ are \emph{concordant} in $\sigma, \tau$ if $\sigma(i) < \sigma(j)$ and $\tau(i) < \tau(j)$ or the same holds with $>$). In particular, 
			\begin{equation*}
				\tau = \frac{|\{\text{concordant pairs}\}|-|\{\text{discordant pairs}\}|}{\binom{n}{2}}
			\end{equation*}
			Kendall's $\tau$ varies in the interval $[-1,+1]$: $\tau = 1$ if $\sigma$ and $\tau$ agree perfectly while $\tau = -1$ if one ordering is the reverse of the other. Furthermore, if $\sigma$ and $\nu$ are independent then $\tau \approx 0$.\\
			The \emph{Spearmann score} is defined as the Pearson correlation coefficient and in the case that the $n$ ranks are distinct integers, it can be computed using the formula
			\begin{equation*}
				r = 1-\frac{6\sum_{i=1}^n (\sigma(i)-\tau(i))^2}{n(n^2-1)}.
			\end{equation*}
			
			As a drawback, correlation based metrics assign the same importance to the first part as to the last part of the rankings; as they equally evaluate exchanges in the first-ranked items and in the ending part of the ranking, they do not properly fit with evaluating orderings. Both Spearmann's $\rho$ and Kendall's $\tau$ directly penalize swaps of 'further located' elements; Considering two rankings $\sigma,\tau$ that differ for one single swap $(i\ j)$, if $i,j$ are far in the rankings, then they will evaluate their difference as being bigger as if they would be nearer. \par
			Finally, the \emph{Normalized Distance-based Performance Measure} NDPM from~\cite{yao_measuring_1995}: Given $\sigma,\tau$ and the following quantities
			\begin{align*}
				&C_+ = \sum_{i,j} \text{sgn}(\sigma{(i)}-\sigma{(j)})\text{sgn}(\tau{(i)} - \tau{(j)})\\
				&C_- = \sum_{i,j} \text{sgn}(\sigma{(i)}-\sigma{(j)})\text{sgn}(\tau{(j)} - \tau{(i)})\\
				&C_u = \sum_{i,j} \left[\text{sgn}(\sigma{(i)}-\sigma{(j)})\right]^2\\
				&C_s = \sum_{i,j} \left[\text{sgn}(\tau{(i)} - \tau{(j)})\right]^2\\
				&C_{u0} = C_u-C_+-C_-
			\end{align*}
		 	from their combination, we get $NDPM(\sigma,\tau) = \frac{C_-+\frac{1}{2}C_{u0}}{C_u}$.
			
			\subsection{Cumulative gain based metrics}\label{app:dcg}
			Constructed with the specific aim of evaluating whether the ordering of relevant elements is respected, they use relevance scores assigned to each element. We assume that the relevance score of an element $i$ is represented by all different relevant scores; in particular, we consider $\text{rel}_i = \sigma(i)$ allowing to fairly compare with other metrics that do not have access to relevance scores of items, but only their position. We initially considered assigning $\text{rel}_i =1 $ to all relevant elements; however, this immediately would imply that DCG and nDCG do not satisfy the IoI property as relevant elements are indistinguishable. The \emph{Discounted cumulative gain DCG} assumes that highly relevant items appearing lower in a search result list should be penalized as the graded relevance value is scaled to be logarithmically proportional to the position of the item; the definition reads $\text{DCG} = \sum_{i=1}^n \frac{\sigma(i)}{\log_2(i+1)}$. The \emph{Normalized discounted cumulative gain nDCG} is a normalization of the DCG through the normalization coefficient IDCG, computed by sorting all elements by their relative relevance, and producing the maximum possible DCG. The two metrics nDCG and DCG are linearly equivalent; Thus, as Figure~\ref{fig:heatmap} already empirically showed, there are no inconsistencies among them. \par
			Strictly connected to the cumulative gain metrics is the \emph{Mean Reciprocal Rank} MRR that evaluates the position of each relevant element in the ranking and computes the average of the reciprocal positional ranking of the results, i.e., $\text{MRR} = \frac{1}{|{R}|}\sum_{i=1}^{|{R}|} \frac{1}{\sigma(i)} $; thus, it is the relevance scores' \emph{harmonic mean}. The \emph{meanRank} is defined as $\text{meanRank}(\sigma) = \frac{1}{|{R}|}\sum_{i=1}^{|{R}|} {\sigma(i)} $ and the \emph{GMR} as the geometric mean of the first $k$ ranked elements of the ranking $\sigma$.

			\subsubsection{Error based metrics} Although meant to evaluate continuous and discrete labels, error-based metrics found application in evaluating rankings. They do not consider the ordering of items but compute the difference in each position, sum it all together, and return an average. We will briefly provide the definitions of each metric in this section. Among them, we find the \emph{mean squared error} MSE, defined as $MSE(\sigma,\tau) = \sum_{i=1}^n (\sigma(i)-\tau(i))^2$ for $\sigma,\tau \in S_n$, and the \emph{mean absolute error} MAE, defined as $MAE(\sigma,\tau) = \sum_{i=1}^n |\sigma(i)-\tau(i)|$. The \emph{rooted mean squared error} RMSE and the \emph{rooted mean absolute error} RMAE are their respective rooted versions, i.e., $RMSE = \sqrt{MSE}$ and $RMAE=\sqrt{MAE}$. The \emph{symmetric mean absolute percentage error} SMAPE is defined as $SMAPE(\sigma, \tau) = \frac{100}{n}\sum_{i=1}^n 2\frac{|\sigma(i)-\tau(i)|}{\sigma(i)+\tau(i)}$   for $\sigma,\tau \in S_n$. Finally, the R$2$ score is defined as $\text{R}2\text{score} = 1- \frac{\sum_{i=1}^n(\sigma(i)-\tau(i))^2}{\sum_{i=1}^n(\sigma(i)-\frac{1}{n}\sum_{j=1}^n\tau(j))^2}$.

\end{document}

%% file: math_commands.tex
\usepackage{amsmath,amsfonts,bm}

\def\eqref#1{equation~\ref{#1}}

\def\1{\bm{1}}

\DeclareMathAlphabet{\mathsfit}{\encodingdefault}{\sfdefault}{m}{sl}
\SetMathAlphabet{\mathsfit}{bold}{\encodingdefault}{\sfdefault}{bx}{n}



%% file: tables/listmetrics.tex
\newcolumntype{G}{>{\raggedright\arraybackslash}m{9cm}}
\newcolumntype{U}{>{\raggedright\arraybackslash}m{4cm}}
\newcolumntype{L}{>{\RaggedLeft\arraybackslash}m{1cm}}

\setlength{\extrarowheight}{2pt}
\begin{table*}[!t]

\begin{tabularx}{\textwidth}{   UG}
    \toprule
    \makecell[l]{Ranking aware \\metrics} &\textbf{nDCG}, \textbf{DCG}, \textcolor{blue}{meanRank}, \textcolor{blue}{GMR}, \textcolor{blue}{MRR} \\ \midrule
    \makecell[l]{Metrics assigning equal \\importance to each position} & \makecell[l]{\textit{SMAPE}, \textit{MAPE}, \textit{MAE}, \textit{RMSE}, \textit{MSE}, \textit{R$^2$ score}, {NDPM}, \\Spearmann $\rho$, Kendall's $\tau$ }\\ \midrule
    Set based metrics & \makecell[l]{\underline{markedness},\underline{ PT}, \underline{ recall}, \underline{LR+}, \underline{Jaccard index}, \underline{F1 score}, \underline{FDR}, \\\underline{accuracy}, \underline{MCC}, \underline{TNR}, \underline{fallout}, \underline{FNR}, \underline{LR-} \underline{informedness}, \\\underline{NPV}, \underline{FOR}, \underline{BA}, \underline{FM}, \underline{precision} } \\ \bottomrule
    \end{tabularx}
    \caption{\label{tab:listmetrics}List of considered metrics; bold, italic, underlined, and plain text indicate \textbf{CGB}, \emph{EB}, \underline{CMB}, and CB metrics. Other metrics are blue color-coded.}

\end{table*}

%% file: tables/identityofindiscrnibles.tex
\newcolumntype{H}{>{\centering\arraybackslash}m{1.2cm}}
\newcolumntype{Y}{>{\centering\arraybackslash}m{0.5cm}}

  \renewcommand{\arraystretch}{1.15}
\let\svtabcolsep\tabcolsep
\setlength\tabcolsep{0pt}

\setlength{\extrarowheight}{2pt}
\let\svtabcolsep\tabcolsep
\setlength\tabcolsep{0pt}
\begin{table*}[!t]
\centering

\begin{tabular}{HHHHH|YYYYYYYYYYYYYYYYYYYY|YYYYYY|YY|YY|YYYY|}
        \rotatebox{90}{ranking length} 
        & \rotatebox{90}{relevant}
        & \rotatebox{90}{baseline}
        & \rotatebox{90}{$\sigma$}
        & \rotatebox{90}{$\tau$}
        & \rotatebox{90}{CMB metrics}
        %
        & \rotatebox{90}{MSE}
        & \rotatebox{90}{RMSE}
        & \rotatebox{90}{MAE}
        & \rotatebox{90}{MAPE}
        & \rotatebox{90}{SMAPE}
        & \rotatebox{90}{R$^2$ score}
        %
        & \rotatebox{90}{Kendall's $\tau$}
        & \rotatebox{90}{Spearmann $\rho$}
        & \rotatebox{90}{DCG}
        & \rotatebox{90}{nDCG}
        & \rotatebox{90}{MRR}
        & \rotatebox{90}{GMR}
        & \rotatebox{90}{NDPM}
        & \rotatebox{90}{meanRank}
\\\hline 
    \centering
    10 & 5 & \text{id} & $(1\ 2)$ & \text{id} & $\circ$ & $\bullet$ & $\bullet$ & $\bullet$ & $\bullet$ & $\bullet$ & $\bullet$ & $\bullet$ & $\bullet$ & $\bullet$ & $\bullet$ & $\bullet$ & $\bullet$ & $\bullet$ & $\bullet$\\
    10 & 5 & \text{id} & $(1\ 2)$ & $(3\ 4)$ & $\circ$ & $\circ$ & $\circ$ & $\circ$ & $\bullet$ & $\bullet$ & $\circ$ & $\circ$ & $\circ$ & $\bullet$ & $\bullet$ & $\circ$ & $\circ$ & $\circ$ & $\circ$\\
    10 & 5 & \text{id} & $(1\ 2)$ & $(2\ 4)$ & $\circ$ & $\bullet$ & $\bullet$ & $\bullet$ & $\circ$ & $\circ$ & $\bullet$ & $\bullet$ &$\bullet$ & $\bullet$ & $\bullet$ & $\circ$ & $\circ$ & $\bullet$ &$\circ$ \\
    \end{tabular}
    \caption{\label{tab:identityofindiscrnibles}Examples of rankings that metrics cannot distinguish. We compare for each evaluation metric $\rem$ the values $\rem(\text{id}, \sigma)$ and $\rem(\text{id}, \tau)$. If the metric fails in distinguishing the two rankings, we impute a $\circ$; else, a $\bullet$.}
\end{table*}

%% file: tables/comparison.tex
\setlength{\extrarowheight}{2pt}
\newcolumntype{S}{>{\raggedright\arraybackslash}m{0.37cm}}
\newcolumntype{K}{>{\raggedright\arraybackslash}m{1.7cm}}

\begin{table*}[!t]

\centering
\begin{tabular}{KSSSSSSSSSSSSSSSSSSSSSSSSSSSSSSSSSSSSSS|YYYYYY|YY|YY|YYYY|}
        & \rotatebox{90}{recall}
        & \rotatebox{90}{FNR}
        & \rotatebox{90}{fallout}
        & \rotatebox{90}{TNR}
        & \rotatebox{90}{precision}
        & \rotatebox{90}{FDR}
        & \rotatebox{90}{NPV}
        & \rotatebox{90}{FOR}
        & \rotatebox{90}{accuracy}
        & \rotatebox{90}{BA}
        & \rotatebox{90}{F1 score}
        & \rotatebox{90}{FM}
        & \rotatebox{90}{MCC}
        & \rotatebox{90}{Jaccard index}
        & \rotatebox{90}{markedness}
        & \rotatebox{90}{LR-}
        & \rotatebox{90}{informedness}
        & \rotatebox{90}{PT}
        & \rotatebox{90}{LR+}  
        & \rotatebox{90}{MSE}
        & \rotatebox{90}{RMSE}
        & \rotatebox{90}{MAE}
        & \rotatebox{90}{MAPE}
        & \rotatebox{90}{SMAPE}
        & \rotatebox{90}{R$^2$ score}
        & \rotatebox{90}{Kendall's $\tau$}
        & \rotatebox{90}{Spearmann $\rho$}
        & \rotatebox{90}{NDPM}
        & \rotatebox{90}{DCG}
        & \rotatebox{90}{nDCG}

        & \rotatebox{90}{MRR}
        & \rotatebox{90}{GMR}
        & \rotatebox{90}{meanRank} 
\\\hline 
    {id. indisc. \qquad} & \xmark & \xmark & \xmark & \xmark & \xmark & \xmark & \xmark & \xmark & \xmark & \xmark & \xmark & \xmark & \xmark & \xmark & \xmark & \xmark & \xmark & \xmark &\xmark & \xmark & \xmark & \xmark & \xmark & \xmark & \xmark & \xmark & \xmark & \xmark & \cmark & \cmark & \xmark & \xmark &  \xmark \\
    {symmetry \qquad} & \cmark & \cmark & \cmark & \cmark & \cmark & \cmark & \cmark & \cmark & \cmark & \cmark & \cmark & \cmark & \cmark & \cmark & \cmark & \cmark & \cmark & \cmark &  \cmark &\cmark & \cmark & \cmark & \xmark & \xmark & \cmark & \cmark & \cmark &  \cmark &\xmark & \xmark & \xmark &\xmark & \xmark\\ 
    rob I & \cmark  &\cmark  &\cmark  &\cmark  &\cmark  &\cmark  &\cmark  &\cmark  &\cmark  &\cmark  &\cmark  &\cmark  &\cmark  &\cmark  &\cmark  &\cmark  &\cmark  &\xmark &\xmark & \xmark & \xmark & \xmark & \xmark & \xmark &\xmark &\cmark  &\xmark &\cmark  & \xmark &\cmark  &\cmark  & \xmark & \xmark\\

    {rob II } & \xmark & \xmark & \xmark & \xmark & \xmark & \xmark & \xmark & \xmark & \xmark & \xmark & \xmark & \xmark & \xmark & \xmark & \xmark & \xmark & \xmark & \xmark &  \xmark & \cmark & \cmark & \cmark & \cmark & \xmark & \cmark & \cmark & \cmark &  \xmark &\xmark & \xmark & \xmark & \xmark & \xmark \\
    {WSD } & \xmark & \xmark & \xmark & \xmark & \xmark & \xmark & \xmark & \xmark & \xmark & \xmark & \xmark & \xmark & \xmark & \xmark & \xmark & \xmark &  \xmark & \xmark &\xmark  & \xmark & \xmark & \xmark & \xmark & \xmark & \xmark & \cmark & \cmark&  \cmark  & \xmark & \xmark & \xmark & \xmark & \xmark  \\
    {sensitivity } & \xmark & \xmark & \xmark & \xmark & \xmark & \xmark & \xmark & \xmark & \xmark & \xmark & \xmark & \xmark & \xmark & \xmark & \xmark & \xmark & \xmark & \xmark &  \xmark & \xmark & \xmark & \xmark & \cmark & \cmark & \xmark & \xmark & \xmark &  \cmark&\cmark & \cmark & \cmark & \cmark & \cmark    \\
    {stability} & \cmark & \cmark & \xmark & \xmark & \cmark & \cmark & \xmark & \xmark & \xmark & \xmark & \cmark & \cmark & \cmark & \xmark & \cmark& \xmark & \cmark & \cmark & \cmark & \cmark & \cmark & \cmark & \cmark & \cmark & \cmark & \cmark & \cmark & \cmark & \cmark & \cmark & \cmark & \xmark & \xmark \\
    {distance} & \xmark & \xmark & \xmark & \xmark & \xmark & \xmark & \xmark & \xmark & \xmark & \xmark & \xmark & \xmark & \xmark & \xmark & \xmark & \xmark & \xmark & \xmark &  \xmark & \xmark & \cmark & \cmark & \xmark & \xmark & \xmark & \cmark & \xmark &  \xmark & \cmark & \cmark & \xmark & \xmark & \xmark \\
    
\end{tabular}
\caption{\label{tab:comparison}Summary table of the property satisfied by the metrics.}

\end{table*}

%% file: tables/summary.tex
\newcolumntype{P}{>{\raggedright\arraybackslash}m{2.2cm}}
\newcolumntype{T}{>{\raggedright\arraybackslash}m{6cm}}
\setlength{\extrarowheight}{5pt}

\begin{table*}[!t]
\centering
\begin{tabularx}{\textwidth}{PTTTTTT}
        & description & domain
        \\\midrule 
        \makecell[l]{identity of \\indiscernibles} & \makecell[l]{in highly sensitive evaluations, \\where detecting tiny differences\\among rankings is essential} & \makecell[l]{(Fair) rank aggregation \\ Recommender Systems \\ Feature ranking/selection} \\\midrule
        {symmetry} & \makecell[l]{ensures that the input rankings\\ have an equal role in the evaluations} & \makecell[l]{Rank aggregation \\
        Contexts independent from ground truths} \\\midrule
        {robustness I} & \makecell[l]{ensures that small changes \\influence proportionately\\ the evaluations } & \makecell[l]{Information Retrieval \\ Rank aggregation} \\\midrule
        {robustness II} & \makecell[l]{ensures independence from \\items renaming} & \makecell[l]{Information Retrieval \\ Rank aggregation \\Feature ranking \\Rank aggregation} \\\midrule
        {sensitivity } & \makecell[l]{for not fully explored rankings, \\ when the interest is on \\to the top part of the rankings} & \makecell[l]{Information Retrieval \\ Recommender Systems}\\\midrule
        {stability} & \makecell[l]{ensures trustworthiness\\ in evaluations $@k$} & \makecell[l]{Information Retrieval \\ Recommender Systems}\\\midrule
        {distance} & \makecell[l]{ensures that the metric in questions \\respect the definition\\ of distance on $S_n$} &\makecell[l]{(Fair) rank aggregation}
\end{tabularx}
\caption{\label{tab:summary} Summary of the properties}
\end{table*}

%% file: main.bbl
\begin{thebibliography}{10}

\bibitem{adomavicius_toward_2005}
Gediminas Adomavicius and Alexander Tuzhilin.
\newblock Toward the next generation of recommender systems: A survey of the
  state-of-the-art and possible extensions.
\newblock {\em IEEE transactions on knowledge and data engineering}, 17(6),
  2005.

\bibitem{amigo_axiomatic_2018}
Enrique Amig{\'o}, Damiano Spina, and Jorge Carrillo-de Albornoz.
\newblock An axiomatic analysis of diversity evaluation metrics: Introducing
  the rank-biased utility metric.
\newblock In {\em SIGIR}, 2018.

\bibitem{buckley_retrieval_2004}
Chris Buckley and Ellen~M Voorhees.
\newblock Retrieval evaluation with incomplete information.
\newblock In {\em SIGIR}, 2004.

\bibitem{cook1986axiomatic}
Wade~D Cook, Moshe Kress, and Lawrence~M Seiford.
\newblock An axiomatic approach to distance on partial orderings.
\newblock {\em RAIRO-Operations Research}, 20(2), 1986.

\bibitem{diaconis}
Persi Diaconis.
\newblock Group representations in probability and statistics.
\newblock {\em Lecture notes-monograph series}, 11, 1988.

\bibitem{10.1145/371920.372165}
Cynthia Dwork, Ravi Kumar, Moni Naor, and Dandapani Sivakumar.
\newblock Rank aggregation methods for the web.
\newblock In {\em WWW}, 2001.

\bibitem{fligner1986distance}
Michael~A Fligner and Joseph~S Verducci.
\newblock Distance based ranking models.
\newblock {\em Journal of the Royal Statistical Society: Series B
  (Methodological)}, 48(3), 1986.

\bibitem{gosgens_good_2021}
Martijn G{\"o}sgens, Anton Zhiyanov, Aleksey Tikhonov, and Liudmila
  Prokhorenkova.
\newblock Good classification measures and how to find them.
\newblock In {\em NIPS}, 2021.

\bibitem{gosgens_systematic_2021}
Martijn~M G{\"o}sgens, Alexey Tikhonov, and Liudmila Prokhorenkova.
\newblock Systematic analysis of cluster similarity indices: How to validate
  validation measures.
\newblock In {\em ICML}, 2021.

\bibitem{gunawardana_evaluating_2015}
Asela Gunawardana, Guy Shani, and Sivan Yogev.
\newblock Evaluating recommender systems.
\newblock In {\em Recommender systems handbook}. 2012.

\bibitem{hall2018theory}
Marshall Hall.
\newblock {\em The theory of groups}.
\newblock Courier Dover Publications, 2018.

\bibitem{farzad}
Farzad~Farnoud Hassanzadeh and Olgica Milenkovic.
\newblock An axiomatic approach to constructing distances for rank comparison
  and aggregation.
\newblock {\em IEEE Transactions on Information Theory}, 60(10), 2014.

\bibitem{herlocker_evaluating_2004}
Jonathan~L Herlocker, Joseph~A Konstan, Loren~G Terveen, and John~T Riedl.
\newblock Evaluating collaborative filtering recommender systems.
\newblock {\em ACM Transactions on Information Systems (TOIS)}, 22(1), 2004.

\bibitem{hoyt_unified_2022}
Charles~Tapley Hoyt, Max Berrendorf, Mikhail Galkin, Volker Tresp, and
  Benjamin~M Gyori.
\newblock A unified framework for rank-based evaluation metrics for link
  prediction in knowledge graphs.
\newblock {\em arXiv preprint arXiv:2203.07544}, 2022.

\bibitem{jarvelin_cumulated_2002}
Kalervo J{\"a}rvelin and Jaana Kek{\"a}l{\"a}inen.
\newblock Cumulated gain-based evaluation of ir techniques.
\newblock {\em ACM Transactions on Information Systems (TOIS)}, 20(4), 2002.

\bibitem{kendall_new_1938}
Maurice~G Kendall.
\newblock A new measure of rank correlation.
\newblock {\em Biometrika}, 30(1/2), 1938.

\bibitem{khaire2022stability}
Utkarsh~Mahadeo Khaire and R~Dhanalakshmi.
\newblock Stability of feature selection algorithm: A review.
\newblock {\em Journal of King Saud University-Computer and Information
  Sciences}, 34(4), 2022.

\bibitem{lin2010rank}
Shili Lin.
\newblock Rank aggregation methods.
\newblock {\em Wiley Interdisciplinary Reviews: Computational Statistics},
  2(5), 2010.

\bibitem{liu2009learning}
Tie-Yan Liu et~al.
\newblock Learning to rank for information retrieval.
\newblock {\em Foundations and Trends{\textregistered} in Information
  Retrieval}, 3(3), 2009.

\bibitem{mogotsi_christopher_2010}
Hinrich Sch{\"u}tze, Christopher~D Manning, and Prabhakar Raghavan.
\newblock {\em Introduction to information retrieval}, volume~39.
\newblock 2008.

\bibitem{sculley2007rank}
D~Sculley.
\newblock Rank aggregation for similar items.
\newblock In {\em SDM}, 2007.

\bibitem{silveira_how_2019}
Thiago Silveira, Min Zhang, Xiao Lin, Yiqun Liu, and Shaoping Ma.
\newblock How good your recommender system is? a survey on evaluations in
  recommendation.
\newblock {\em International Journal of Machine Learning and Cybernetics}, 10,
  2019.

\bibitem{tamm_quality_2021}
Yan-Martin Tamm, Rinchin Damdinov, and Alexey Vasilev.
\newblock Quality metrics in recommender systems: Do we calculate metrics
  consistently?
\newblock In {\em RecSys}, 2021.

\bibitem{valcarce_robustness_2018}
Daniel Valcarce, Alejandro Bellog{\'\i}n, Javier Parapar, and Pablo Castells.
\newblock On the robustness and discriminative power of information retrieval
  metrics for top-n recommendation.
\newblock In {\em RecSys}, 2018.

\bibitem{valcarce_assessing_2020}
Daniel Valcarce, Alejandro Bellog{\'\i}n, Javier Parapar, and Pablo Castells.
\newblock Assessing ranking metrics in top-n recommendation.
\newblock {\em Information Retrieval Journal}, 23, 2020.

\bibitem{yao_measuring_1995}
YY~Yao.
\newblock Measuring retrieval effectiveness based on user preference of
  documents.
\newblock {\em Journal of the American Society for Information science}, 46(2),
  1995.

\end{thebibliography}
